\documentclass[11pt]{article}

\usepackage{amsmath, amssymb, amsthm, subfigure}

\usepackage{times}
\usepackage{bm}
\usepackage{natbib}
\usepackage{mhequ}
\usepackage[pdftex]{graphicx}
\usepackage{authblk}
\usepackage{fullpage}

\usepackage[width=.9\textwidth]{caption}

\usepackage[plain,noend]{algorithm2e}

\def \be{\begin{equs}}
\def \ee{\end{equs}}
\def \dat{\bm{y}_{\mathrm{obs}}} 
\def \datt{\bm{y}_{\bt'}} 

\newcommand{\iidsim}{\stackrel{\mbox{\tiny iid}}{\sim}}

\def\be{\begin{equs}}
\def \ee{\end{equs}}
\newtheorem{thm}{Theorem}
\newtheorem{lemma}[thm]{Lemma}
\newtheorem{cor}[thm]{Corollary}

\newtheorem{prop}[thm]{Proposition}

\def \bt{\bm{\theta}}

\def \E{\mathbb{E}}
\def \by{\bm{y}}

\def \cX{\mathcal{X}}

\title{The Use of a Single Pseudo-Sample \\in Approximate Bayesian Computation}

\author[1,2]{Luke Bornn\thanks{bornn@stat.harvard.edu}}
\author[2]{Natesh S. Pillai}
\author[3]{Aaron Smith}
\author[4]{Dawn Woodard}
\affil[1]{Department of Statistics and Actuarial Science, Simon Fraser University}
\affil[2]{Department of Statistics, Harvard University}
\affil[3]{Department of Mathematics and Statistics, University of Ottawa}
\affil[4]{School of Operations Research and Information Engineering, Cornell University}

\begin{document}

\maketitle




\begin{abstract}
We analyze the computational efficiency of approximate Bayesian
computation (ABC), which approximates a likelihood function by drawing
pseudo-samples from the associated model. For the rejection sampling
version of ABC, it is known that multiple pseudo-samples cannot
substantially increase (and can substantially decrease) the efficiency
of the algorithm as compared to employing a high-variance estimate
based on a single pseudo-sample.  We show that this conclusion also
holds for a Markov chain Monte Carlo version of ABC,  implying
that it is unnecessary to tune the number of pseudo-samples used in
ABC-MCMC. This conclusion is in contrast to particle MCMC methods, for
which increasing the number of particles can provide large gains in
computational efficiency.
\end{abstract}

\section{Introduction}

Approximate Bayesian computation (ABC) is a family of algorithms for
Bayesian inference that address situations where the likelihood
function is intractable to evaluate but where one can obtain samples
from the model.  These methods are now widely used in population
genetics, systems biology, ecology, and other areas, and are
implemented in an array of popular software packages \citep{tavare1997inferring,marin2012approximate}.  Let $\bt \in
\Theta$ be the parameter of interest with prior density
$\pi(\bm{\theta})$, $\dat \in \mathcal{Y}$ be the observed data and
$p(\bm{y}|\bt)$ be the model.  The simplest ABC
algorithm first generates a sample $\bm{\theta'} \sim
\pi(\bm{\theta})$ from the prior, then generates a {\it
  pseudo-sample} $\bm{y}_{\bt'} \sim p( \cdot | \bm{\theta}')$.
Conditional on $\bm{y}_{\bt'} = \dat$, the distribution of
$\bm{\theta'}$ is the posterior distribution $\pi( \bt | \dat)\propto
\pi(\bt) p(\dat | \bt)$.  For all but the most trivial discrete
problems, the probability that $\bm{y}_{\bt'} = \dat$ is either zero
or very small. Thus the condition of exact matching of pseudo-data to
the observed data is typically relaxed to $\|\eta(\dat) -
\eta(\bm{y}_{\bt'})\| < \epsilon,$ where $\eta$ is a low-dimensional
summary statistic, $\|\cdot\|$ is a distance function, and
$\epsilon>0$ is a tolerance level.  The resulting algorithm gives
samples from the target density $\pi_\epsilon$ that is proportional to
$\left[ \pi(\bt)\int \mathbf{1}_{\{\| \eta(\dat) -
    \eta(\bm{y})\|<\epsilon\}}p(\bm{y}|\bt) d\bm{y}\right]$.  If the
tolerance $\epsilon$ is small enough and the statistic(s) $\eta$ good
enough, then $\pi_{\epsilon}$ can be a good approximation to
$\pi(\bt|\dat)$.

A generalized version of this method \citep{wilk:08} is given in
Algorithm \ref{abc-gen}, where $K(\eta)$ is an unnormalized
probability density function, $M$ is an arbitrary number of
pseudo-samples, and $c$ is a constant satisfying $c \geq \sup_\eta
K(\eta)$.

\begin{algorithm}
  \For{$t=1$ to $n$}{
    \Repeat{$u <  \frac{1}{cM} \sum_{i=1}^M K(\eta(\dat) -\eta(\bm{y}_{i,\bt'}))$}{
    Generate $\bm{\theta}' \sim \pi(\bm{\theta})$,
     $\bm{y}_{i,\bt'} \iidsim p( \bm{y} | \bt')$ for $i=1,\dots, M$,
    and $u \sim \mbox{Uniform}[0,1]$.
    }
    Set $\bt^{(t)} = \bt'$\;
  }
\caption{Generalized ABC\label{abc-gen}}
\end{algorithm}

Using the ``uniform kernel'' $K(\eta) = \mathbf{1}_{\{\|\eta\| <
  \epsilon\}}$ and taking $M= c=1$ yields the version described
  above.  Algorithm~\ref{abc-gen} yields samples $\bt^{(t)}$ from the
 kernel-smoothed posterior density
\be
\label{eqn:ker} \pi_K(\bt | \dat) \propto \pi(\bt) \int K\left(
\eta(\dat)- \eta(\by) \right)\, p(\by|\bt)\, d\by.  \ee Although
Algorithm~\ref{abc-gen} is nearly always presented in the special case
with $M=1$ \citep{wilk:08}, it is easily verified that $M > 1$
still yields samples from $\pi_K(\bt | \dat)$.  


Many variants of Algorithm~\ref{abc-gen} exist. For any unbiased, nonnegative estimator $\hat{\mu}(\bt)$ of $\pi_K( \bt | \dat)$ and any reversible transition kernel $q$ on $\Theta$,  Algorithm
\ref{pseudomarg} gives a Markov chain Monte Carlo (MCMC) version that is based on the pseudo-marginal algorithm of \citep{andrieu2009the-pseudo-marginal} (see also \citep{marj:03,wilk:08,Andrieu2014establishing-some-order}). To obtain an algorithm that is similar to Algorithm~\ref{abc-gen}, a good choice for $\hat{\mu}(\bt)$ is \be
 \label{eqn:KDE} \hat{\pi}_{K,M}(\bt| \dat) \equiv \pi(\bt){1 \over
 M} \sum_{i=1}^M K\left(\eta(\dat) - \eta(\by_{i,\bt})\right), \ee where $M \in \mathbb{N}$. When we refer to Algorithm \ref{pseudomarg}, we always use this family of weights unless stated otherwise. Despite the fact that the expression (\ref{eqn:KDE}) depends on $M$, as $n\rightarrow\infty$ the  distribution of $x^{(n)}$ in Algorithm~\ref{pseudomarg} converges to the same distribution $\pi_K$ under mild conditions  \citep{Andrieu2014establishing-some-order}. Although the choice of $M$ in \eqref{eqn:KDE} generally does not affect the limiting distribution of $x^{(n)}$, it does affect the evolution of the stochastic process $\{ x^{(n)} \}_{n \in \mathbb{N}}$. 

\begin{algorithm}
  Initialize $x^{(0)}$ and generate $T^{(0)} = \hat{\mu}(x^{(0)})$\;
  \For{$t=1$ to $n$}{
    Generate $x'\sim q(\cdot \vert x^{(t-1)})$,
     $T' = \hat{\mu}(x')$,
    and $u \sim \mbox{Uniform}[0,1]$\;
    \eIf{$u \leq \frac{T' q(x^{(t-1)}| x')}{T^{(t-1)} q(x'| x^{(t-1)})}$}{
      Set $\left( x^{(t)}, T^{(t)} \right)$ = $\left( x', T' \right)$\;
    }{
      Set $\left( x^{(t)}, T^{(t)} \right)$ = $\left( x^{(t-1)}, T^{(t-1)} \right)$\;
    }
  }
\caption{Pseudo-Marginal MCMC\label{pseudomarg}}
\end{algorithm}

We address the effect of $M$ on the efficiency of
Algorithm~\ref{pseudomarg} when using the weight \eqref{eqn:KDE}.  Increasing the number of pseudo-samples $M$
decreases the variance of \ref{pseudomarg}, which one might
think could improve the efficiency of Algorithm
\ref{pseudomarg}.  Indeed, increasing $M$ does decrease the asymptotic variance of the associated Monte Carlo estimates (\cite{Andrieu2014establishing-some-order}).  However, increasing $M$ also increases the computational cost of each step of the algorithm.  Although the tradeoff between these two factors is quite complicated, our main result in this paper gives a good compromise: the choice $M=1$ yields a running time within a factor of two of optimal.  We use natural definitions of running time
in the contexts of serial and parallel computing, which are extensions
of those used by \cite{pitt2012on-some},
\cite{sherlock2013on-the-efficiency}, and
\cite{Doucet2014efficient-implementation}, and which capture the time
required to obtain a particular Monte Carlo accuracy. Our definition
in the serial computing case is the number of pseudo-samples
generated, which is an appropriate measure of running time when
drawing pseudo-samples is more computationally intensive than the
other steps in Algorithm~\ref{pseudomarg}, and when the
expected computational cost of drawing each pseudo-sample is the same,
i.e.\ when there is no computational discounting due to having
pre-computed relevant quantities.


Several authors have drawn the conclusion that in many situations the approximately optimal value of $M$ in {\it pseudo-marginal} algorithms (a class of algorithms that includes Algorithm~\ref{pseudomarg}) is obtained by tuning
$M$ to achieve a particular variance for the estimator
$\hat{\pi}_{K,M}(\bt| \dat)$ \citep{pitt2012on-some,
  sherlock2013on-the-efficiency,Doucet2014efficient-implementation}. This often means a value of $M$ that is much larger than 1. We demonstrate that in Algorithm~\ref{pseudomarg} such a tuning process is
unnecessary, since near-optimal efficiency is obtained by using
estimates based on a single pseudo-sample
(Proposition~\ref{thm:unifcor} and Corollary~\ref{cor:M1best}); these estimates are lower-cost and higher-variance than those based on several pseudo-samples.  This result assumes only that the kernel $K(\eta)$ is the uniform kernel
$\mathbf{1}_{\{\|\eta\| < \epsilon\}}$ (the most common choice). In
particular, and in contrast to earlier work, it does not make any
assumptions about the target distribution $\pi(\bt | \dat)$.

Our result is in contrast to particle MCMC methods
\citep{andrieu2010particle}, which 
\citet{flury2011bayesian} demonstrated could require thousands, millions or more
particles to obtain sufficient accuracy in some realistic problems.  This difference between particle MCMC and
ABC-MCMC is largely due to the interacting nature of the particles in particle MCMC, allowing for better path samples.

\section{Efficiency of ABC and ABC-MCMC}\label{sec:mcmc}
 
For a measure $\mu$ on space $\cX$ let $\mu(f) \equiv \int f(x)
\mu(dx)$ be the expectation of a real-valued function $f$ with respect
to $\mu$ and let $L^2(\mu) =\{f: \mu(f^2) < \infty\}$ be the space of
functions with finite variance. For any reversible Markov transition kernel $H$ with stationary distribution $\mu$, any
function $f\in L^2(\mu)$, and Markov chain $X^{(t)}$ evolving
according to $H$, the MCMC estimator of $\mu(f)$ is $\bar{f}_n \equiv \frac{1}{n} \sum_{t=1}^{n} f(X^{(t)})$. The error of this estimator can be measured by the \textit{asymptotic variance}:
\be \label{eqn:nav} v(f,H) =
\lim_{n \rightarrow \infty} n \mathrm{Var}_{H} \left( \bar{f}_n
\right) \ee 
which is closely related to the autocorrelation of the
samples $f(X^{(t)})$ \citep{tierney1998ordering}.

If $H$ is {\it geometrically ergodic}, $v(f,H)$ is guaranteed to be
finite for all $f \in L^2(\mu)$, while in the non-geometrically
ergodic case $v(f, H)$ may or may not be finite \citep{robe:rose:08}.
When $v(f, H)$ is infinite our results still hold but
are not informative.  The fact that our results do not require geometric
ergodicity distinguishes them from many results on efficiency of MCMC
methods \citep{guan:kron:07,wood:schm:hube:09a}. When $X^{(t)}\iidsim \mu$, we define $v(f) = \mathrm{Var} \left( f(X^{(t)}) \right)$.

We will describe the running time of Algorithms~\ref{abc-gen} and
\ref{pseudomarg} in two cases: first, when the computations are done
serially, and second, when they are done in parallel across $M$
processors.  Using \eqref{eqn:nav}, the variance of $\bar{f}_n$ from a
single (reversible) Markov chain $H$ of length $n$ is roughly $v(f,H)/n$, so to
achieve variance $\delta >0$ in the serial context we need
$v(f,H)/\delta$ iterations.  Similarly, the variance of $\bar{f}_{n}$ from a collection of $M$ reversible Markov chains, each run for $n$ steps, is roughly $v(f,H)/(Mn)$. Thus, to achieve variance $\delta >0$ in the parallel context we need to run each Markov chain for $n > v(f,H)/(\delta M)$ iterations. 

Although our definitions make sense for any function $f$ of the two values $( \bm{\theta}^{(t)}, T^{(t)} )$ described by Algorithm \ref{pseudomarg}, throughout the rest of the note we restrict our attention to functions that depend only on the first coordinate, $\bm{\theta}^{(t)}$. That is, when discussing Algorithm \ref{pseudomarg} or other pseudo-marginal algorithms, we restrict our attention to functions that satisfy $f(\bm{\theta},T_{1}) = f(\bm{\theta},T_{2})$ for all $\bm{\theta}$ and all $T_{1}, T_{2}$. We slightly abuse this in our notation,   not distinguishing between a function $f: \Theta \mapsto \mathbb{R}$ of a single variable $\bm{\theta}$ and a function $f: \Theta \times \mathbb{R}^{+} \mapsto \mathbb{R}$ of two variables $(\bm{\theta}, T)$ that only depends on the first coordinate.

Let $Q_M$ be the transition kernel of Algorithm~\ref{pseudomarg}; like all
pseudo-marginal algorithms, $Q_M$ is reversible
\citep{andrieu2009the-pseudo-marginal}.  We assume that drawing
pseudo-samples is the slowest part of the computation, and that
drawing each pseudo-sample takes the same amount of time on average
(as also assumed in Pitt et al.\ 2012\nocite{pitt2012on-some},
Sherlock et al.\ 2013\nocite{sherlock2013on-the-efficiency}, Doucet et
al.\ 2014\nocite{Doucet2014efficient-implementation}).  Then the
running time of $Q_M$ in the serial context can be measured as the
number of iterations times the number of pseudo-samples drawn in each
iteration, namely $C_{f,Q_M,\delta}^{\mbox{\tiny ser}} \equiv M
v(f,Q_M)/\delta$.  In the context of parallel computation across $M$
processors, we compare two competing approaches that utilize all the
processors.  These are: (a) a single chain with transition kernel
$Q_M$, where the $M>1$ pseudo-samples in each iteration are drawn
independently across $M$ processors; and (b) $M$ parallel chains with
transition kernel $Q_1$.  The running time of these methods can be
measured as the number of required Markov chain iterations to obtain
accuracy $\delta$, namely $C_{f,Q_M,\delta}^{\mbox{\tiny M-par}}
\equiv v(f,Q_M)/\delta$ utilizing method (a) and
$C_{f,Q_1,\delta}^{\mbox{\tiny M-par}} \equiv v(f,Q_1)/(\delta M)$
utilizing method (b). Since both measures of computation time are based on the \textit{asymptotic} variance of the underlying Markov chain, they ignore the initial `burn-in' period and are most appropriate when the desired error $\delta$ is small. Other measures of computation time should be used if the Markov chains are being used to get only a very rough picture of the posterior (\textit{e.g.} to locate, but not explore, a single posterior mode). Note, however, that in practice there is typically no burn-in period for
Algorithm~\ref{pseudomarg}, since it is usually initialized using samples from ABC \citep{marin2012approximate}.

For Algorithm~\ref{abc-gen} the running time, denoted by $R_M$, is
defined analogously.  However, we must account for the fact that each
iteration of $R_M$ yields one accepted value of $\bt$, which may
require multiple proposed values of $\bt$ (along with the associated
computations, including drawing pseudo-samples).  The number of
proposed values of $\bt$ to get one acceptance has a geometric
distribution with mean equal to the inverse of the marginal
probability $p_{acc}(R_M)$ of accepting a proposed value of $\bt$.
So, similarly to $Q_M$, the running time of $R_M$ in the context
of serial computing can be measured as $C_{f,R_M,\delta}^{\mbox{\tiny
ser}} \equiv M v(f)/(\delta \, p_{acc}(R_M))$, and the computation
time in the context of parallel computing can be measured as
$C_{f,R_M,\delta}^{\mbox{\tiny M-par}} \equiv v(f)/(\delta \,
p_{acc}(R_M))$ utilizing method (a) and $C_{f,R_1,\delta}^{\mbox{\tiny
M-par}} \equiv v(f)/(\delta M p_{acc}(R_1))$ utilizing method (b).

Using these definitions, we first state the result that $M=1$ is
optimal for ABC (Algorithm 1).  This result is widely known but we
could not locate it in print, so we include it here for completeness.

\begin{lemma}\label{thm:abc}
The marginal acceptance probability of ABC (Algorithm \ref{abc-gen}) does not
depend on $M$.  For $M>1$ the running times
$C_{f,R_M,\delta}^{\mbox{\tiny ser}}$ and
$C_{f,R_M,\delta}^{\mbox{\tiny M-par}}$ of ABC in the serial and
parallel contexts satisfy
$C_{f,R_M,\delta}^{\mbox{\tiny
ser}} = M C_{f,R_1,\delta}^{\mbox{\tiny
ser}}$ and $C_{f,R_M,\delta}^{\mbox{\tiny
M-par}} = M C_{f,R_1,\delta}^{\mbox{\tiny
M-par}}$
for any $f \in L^2(\pi_K)$ and any $\delta>0$.
\end{lemma}
\begin{proof}
The marginal acceptance probability of Algorithm \ref{abc-gen}
is 
\begin{align*}
&\int \pi(\bt) \left[\prod_{i=1}^M p( \bm{y}_{i,\bt} | \bt) \right]
\left[\frac{1}{cM} \sum_{i=1}^M K(\eta(\dat) -
  \eta(\bm{y}_{i,\bt}))\right] \, d \bt \, d \bm{y}_{1,\bt}
\ldots d \bm{y}_{M,\bt}\\
&= \frac{1}{c} \int \pi(\bt) p(\bm{y} | \bt) K(\eta(\dat) - \eta(\bm{y}))\, d \bt \,
d \bm{y}
\end{align*}
which does not depend on $M$.  The results for the running times
follow immediately.
\end{proof}

Our contribution is to show a similar result for ABC-MCMC. A potential concern regarding Algorithm~\ref{pseudomarg} is raised by \cite{lee2013ergodicity}, who point out that this algorithm is generally not geometrically ergodic when $q$ is a local proposal distribution, such as a random walk proposal. This is due to the fact that in the tails of the distribution $\pi_{K}$, the pseudo-data $y_{i,\bt}$ are very different from $y_{\mathrm{obs}}$ and so the proposed moves are rarely accepted. This problem can be fixed in several ways. Lee and Latuszynski (2013)\nocite{lee2013ergodicity} give a sophisticated solution that involves choosing a random number of pseudo-samples at every step of Algorithm~\ref{pseudomarg}, and they show that this modification increases the class of target distributions for which the ABC-MCMC algorithm is geometrically ergodic. One consequence of our Proposition 4 is that a simpler `obvious' fix to the problem of geometric ergodicity does not work: increasing the number of pseudo-samples used in Algorithm~\ref{pseudomarg} from 1 to any fixed number $M$ has no impact on the geometric ergodicity of the algorithm. 


Our main tool in analyzing Algorithm~\ref{pseudomarg} will be the results
of \cite{Andrieu2014establishing-some-order}. Two random variables $X$
and $Y$ are \textit{convex ordered} if $\E[\phi(X)] \leq \E[\phi(Y)]$
for any convex function $\phi$; we denote this relation by $X
\leq_{cx} Y$.  Let $H_{1}, H_{2}$ be the transition kernels of two
pseudo-marginal algorithms with the same proposal kernel $q$ and the
same target marginal distribution $\mu$. Denote by $T_{1,x}$ and
$T_{2,x}$ the estimators of the unnormalized target used by $H_{1}$
and $H_{2}$ respectively. Recall the asymptotic variance $v(f, H)$
from (\ref{eqn:nav}); although $f$ could be a function on
$\mathcal{X}\times \mathbb{R}^+$, we restrict our attention to
functions $f$ on the non-augmented
state space $\mathcal{X}$.  Then if $T_{1,x} \leq_{cx} T_{2,x}$ for all $x$, Theorem 3 of \cite{Andrieu2014establishing-some-order} shows that $v(f, H_1)
\leq v(f, H_2)$ for all $f \in L^{2}(\mu)$. As shown in Section 6 of that work, this tool can be used to show that increasing $M$ decreases the asymptotic variance of Algorithm~\ref{pseudomarg}: 

\begin{cor}\label{thm:likFree}
For any $f \in L^{2}(\pi_K)$ and any $M \leq N \in \mathbb{N}$ we have $ v(f,
Q_M) \geq v( f, Q_{N}).$
\end{cor}

\noindent In the appendix we give a similar comparison result for the
alternative version of ABC-MCMC described in \cite{wilk:08}.

We will show that, despite Corollary~\ref{thm:likFree}, it is not generally an advantage to use a large value of $M$ in Algorithm~\ref{pseudomarg}. To do this, we first give a result that follows easily from Theorem~3 of
\cite{Andrieu2014establishing-some-order}.  For any $ 0 \leq \alpha <
1$ and $i \in \{1,2\}$, define the
estimator
\be \label{eqn:handicap} T_{i,x,\alpha} = \begin{cases} 0
&\mbox{with\, probability}\, \alpha \\ 
T_{i,x}/(1 - \alpha) & \mbox{otherwise.} \end{cases} \ee
$T_{i,x,\alpha}$ has the same mean as $T_{i,x}$, but is a worse estimator. In particular, Proposition 1.2 of \cite{Lekela2014conditional-convex} implies $T_{i,x} \leq_{cx} T_{i,x,\alpha}$ for any $0 \leq \alpha \leq 1$ and any $i \in \{1,2\}$ (Equation \eqref{eqn:handicap} gives the coupling required by that proposition). We have:

\begin{thm} \label{thm:R1RM}
Assume that $H_2$ is nonnegative definite.  For any $ 0 \leq
\alpha < 1$, if $T_{1,x}\leq_{cx} T_{2,x,\alpha}$ for all $x$, then for any $f \in
L^{2}(\mu)$ we have \be v(f, H_1) \leq \frac{1+\alpha}{1-\alpha} v(
f, H_2).  \ee
\end{thm}
\noindent Theorem~\ref{thm:R1RM} is proven in the appendix. The assumption that
$H_2$ is nonnegative definite is a common technical assumption
in analyses of the efficiency of Markov chains
\citep{wood:schm:hube:09a,nara:rakh:10}, and is done here so we can compare $v(f,H_{2})$ to $v(f)$. It can easily be achieved, for example, by incorporating a ``holding probability'' (chance of proposing to stay in the same location) of
1/2 into the proposal kernel $q$ \citep{wood:schm:hube:09a}. 

Theorem~\ref{thm:R1RM} yields the following bound for Algorithm~\ref{pseudomarg}, proven in the appendix.

\begin{prop} \label{thm:unifcor}
If $K(\eta) = \textbf{1}_{\{\| \eta \| < \epsilon\}}$ for some
$\epsilon>0$, and if the transition kernel $Q_M$ of
Algorithm~\ref{pseudomarg} is nonnegative definite, then for any $f \in
L^2(\pi_K)$ we have \be v(f, Q_1)
\leq (2 M - 1) v(f, Q_M). \label{eqn:firstBound}\ee
\end{prop}
%

This yields the following bounds on the running times:

\begin{cor}\label{cor:M1best}
Using the uniform kernel and assuming that $Q_M$ is nonnegative
definite, the running time of
$Q_1$ is at most twice that of $Q_M$, for both serial
and parallel computing:
\begin{align*}
C_{f,Q_1,\delta}^{\mbox{\tiny ser}} \leq
2C_{f,Q_M,\delta}^{\mbox{\tiny ser}} \qquad \mbox{and}
\qquad C_{f,Q_1,\delta}^{\mbox{\tiny M-par}} \leq
2C_{f,Q_M,\delta}^{\mbox{\tiny M-par}}.
\end{align*}

\end{cor}
\noindent Corollary~\ref{cor:M1best} implies that, under the condition that a uniform kernel is being used and that all pseudosamples have the same computational cost, it is only possible to improve the running time of the Markov chain by a factor of 2 by choosing $Q_M$ rather than $Q_1$.  Thus, under these reasonable conditions, there is never a strong reason to use $Q_{M}$ over $Q_{1}$, and in fact there can be a strong reason to use $Q_{1}$ over $Q_{M}$.

\subsection{Simulation study}
We now demonstrate these results through a simple simulation study,
showing that choosing $M>1$ is seldom beneficial. We consider the
model $\by | \bt \sim \mathcal{N}(\bt, \sigma_y^2)$ for a single
observation $\by$, where $\sigma_y$ is known and $\bt$ is given a
standard normal prior, $\bt \sim \mathcal{N}(0, 1)$.
We apply Algorithm~\ref{pseudomarg} using proposal distribution $q(\bt |
\bt^{(t-1)}) = \mathcal{N}(\bt; 0, 1)$, summary statistic $\eta(\by) = \by$, and $K(\eta) =
\mathbf{1}_{\{\|\eta\| < \epsilon\}}$ equal to the uniform kernel with
bandwidth $\epsilon$, and subsequently study the algorithm's acceptance rate, which due to the independent proposal is analogous to effective sample size. 
We start by exploring the case where $\dat=2$ and $\sigma_y=1$,
simulating the Markov chain for $5$ million iterations. Figure
\ref{fig:1} (left) shows the acceptance rate per generated pseudo-sample as a
function of $M$. Large $M$ does not provide a
benefit in terms of accepted $\bt$ samples per generated
pseudo-sample, and can even decrease this measure of efficiency, supporting the result of
Corollary~\ref{cor:M1best}. In fact, replacing the uniform kernel with a Gaussian kernel (not shown) leads to indistinguishable results.
In Figure \ref{fig:1}
(right) we look at the $\epsilon$ which results from requiring that a fixed
percentage ($0.4\%$) of $\bt$ samples are accepted per unit
computational cost.  This
means that for $M=1$ we require that $0.4\%$ of samples are accepted,
while for $M=64$ we require that $0.4\times 64 = 25.6\%$ of samples
are accepted.  

In certain cases, there is an initial fixed computational cost to
generating pseudo-samples, after which generating subsequent
pseudo-samples is computationally inexpensive. If $\by$ is a
length-$n$ sample from a Gaussian process, for instance, there is an
initial $O(n^3)$ cost to decomposing the covariance, after which each
pseudo-sample may be generated at $O(n^2)$ cost. In the case with discount factor $\delta>1$
(representing a $\delta \times$ cost reduction for all pseudo-samples
after the first), the pseudo-sampling cost is $(1+(M-1)/\delta)$, so
we require that $0.4\times (1+(M-1)/\delta)\%$ of $\bt$ samples are
accepted. For example, a discount of $\delta =16$ with $M=64$ requires
that $2.0\%$ of samples are accepted.  Figure \ref{fig:1} shows that,
for $\delta=1$, larger $M$ results in larger $\epsilon$; in other
words, for a fixed computational budget and no discount, $M=1$ gives
the smallest $\epsilon$. For discounts $\delta>1$, however, increasing $M$
can indeed lead to reduced $\epsilon$.
\begin{figure}[h]%
\centering 
{\includegraphics[width=.4\textwidth]{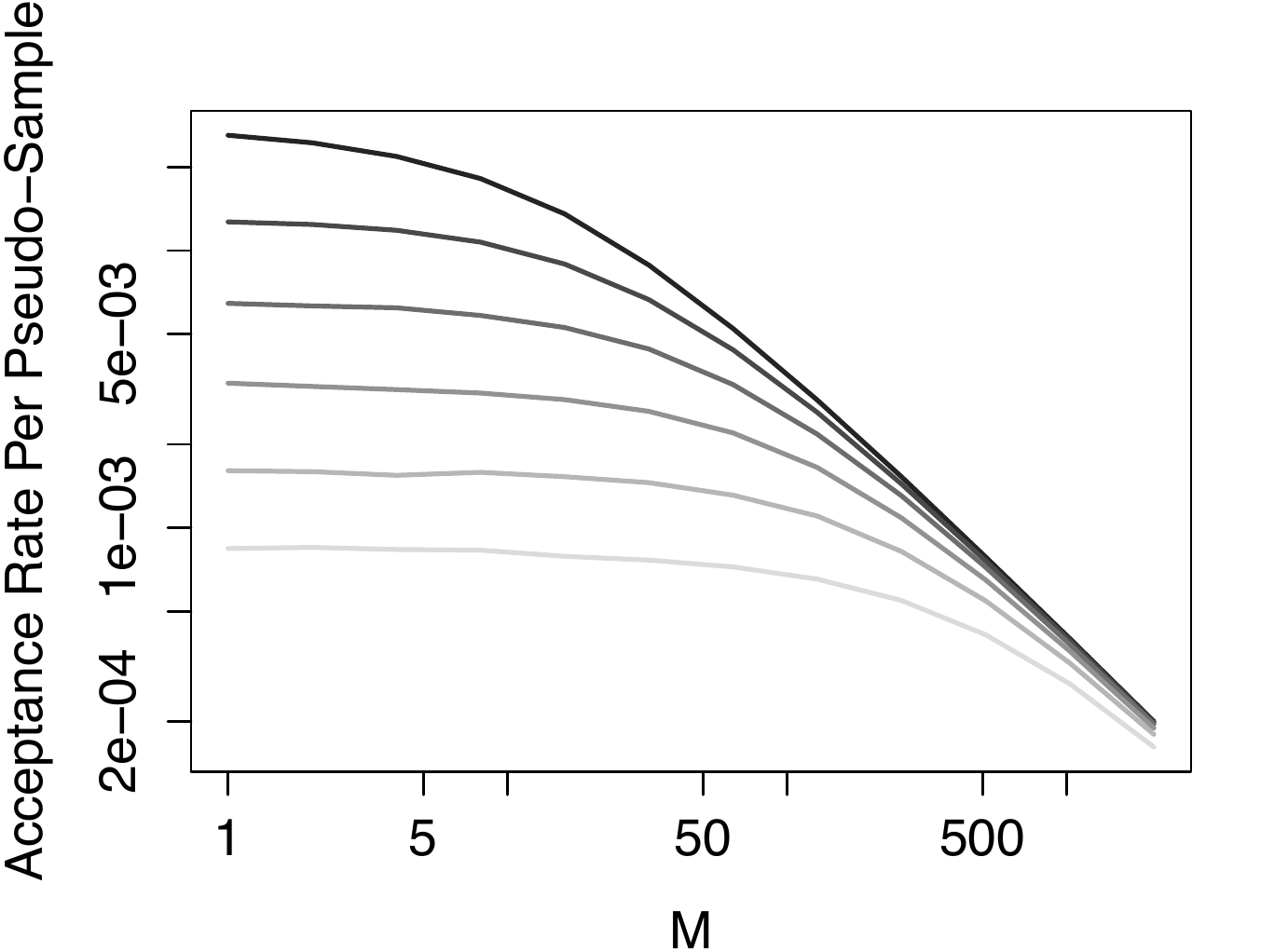}}\qquad
{\includegraphics[width=.4\textwidth]{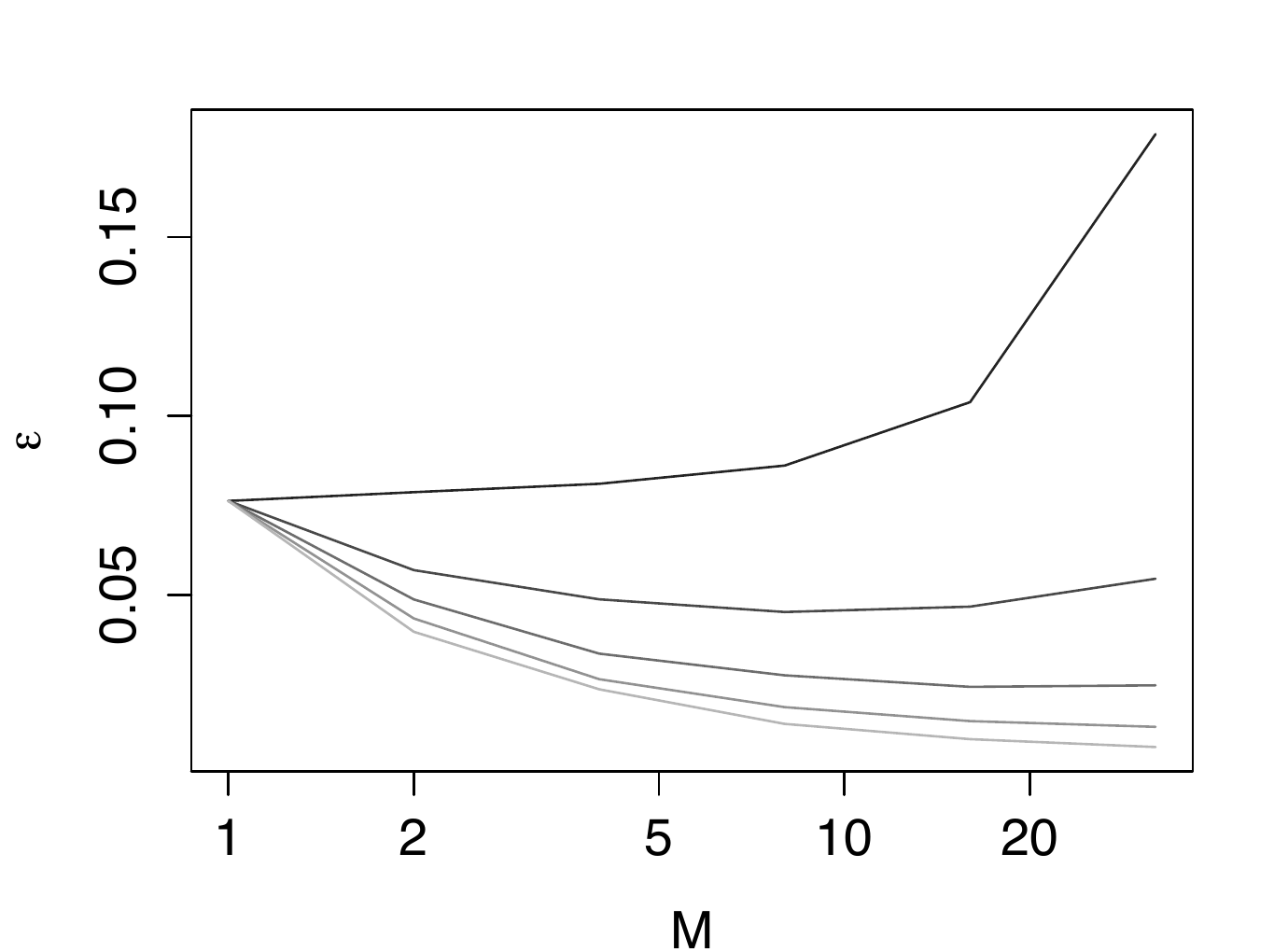}}
\protect\caption{\label{fig:1} Left: Acceptance rate per pseudo-sample. Lines correspond to
$\epsilon=0.5^2$ (black) through $\epsilon=0.5^6$
(light grey), spaced uniformly on log scale.  Right: The $\epsilon$
resulting from requiring that $0.4\%$ of $\bt$ samples are accepted
per unit computational cost.  Lines correspond
to different computational savings for pseudo-samples beyond
the first, in the range $\delta = 1$ (black, no cost savings), $2, 4, 8, 
16$ (light grey).}
\end{figure}

We further explore discounting in Figure \ref{fig:3},
which uses a discount of $\delta = 8$ and varies $\dat$ from $2$ to $8$
(left plot) and $\sigma_y$ from $0.01$ to $2$ (right plot). In both cases the changes are meant to induce a
divergence between the prior and the likelihood, and hence the prior
and posterior. In these figures, the requisite $\epsilon$ are scaled
such that at $M=1$ all normalized $\epsilon$ are $1$.
Figure~\ref{fig:3} (left) shows that as $\dat$ grows, the benefit associated
with using higher values of $M$ shrinks and eventually disappears.
This is because for large $\dat$ a large value of $\epsilon$ is
required in order to frequently get a nonzero value for the
approximated likelihood and thus a reasonable
acceptance rate; for instance, the unnormalized value of $\epsilon$ is
$0.08$ when $\dat=2$ and $M=1$, while $\epsilon=7.64$ when $\dat=8$
and $M=1$. As such, the increased diversity from multiple samples is
dwarfed by the scale of $\epsilon$. In Figure \ref{fig:3} (right) we examine sensitivity of our conclusions to
$\sigma_y$. For large $\sigma_y$, additional (discounted-cost)
pseudo-samples provide a benefit, because they improve the accuracy of
the approximated likelihood. However, for small
values of $\sigma_y$, the variability of the pseudo-samples
$y_{i,\bt'}$ is low and so additional pseudo-samples do not provide
much incremental improvement to the likelihood approximation.  
\begin{figure}[h]%
\centering
\includegraphics[width=.45\textwidth]{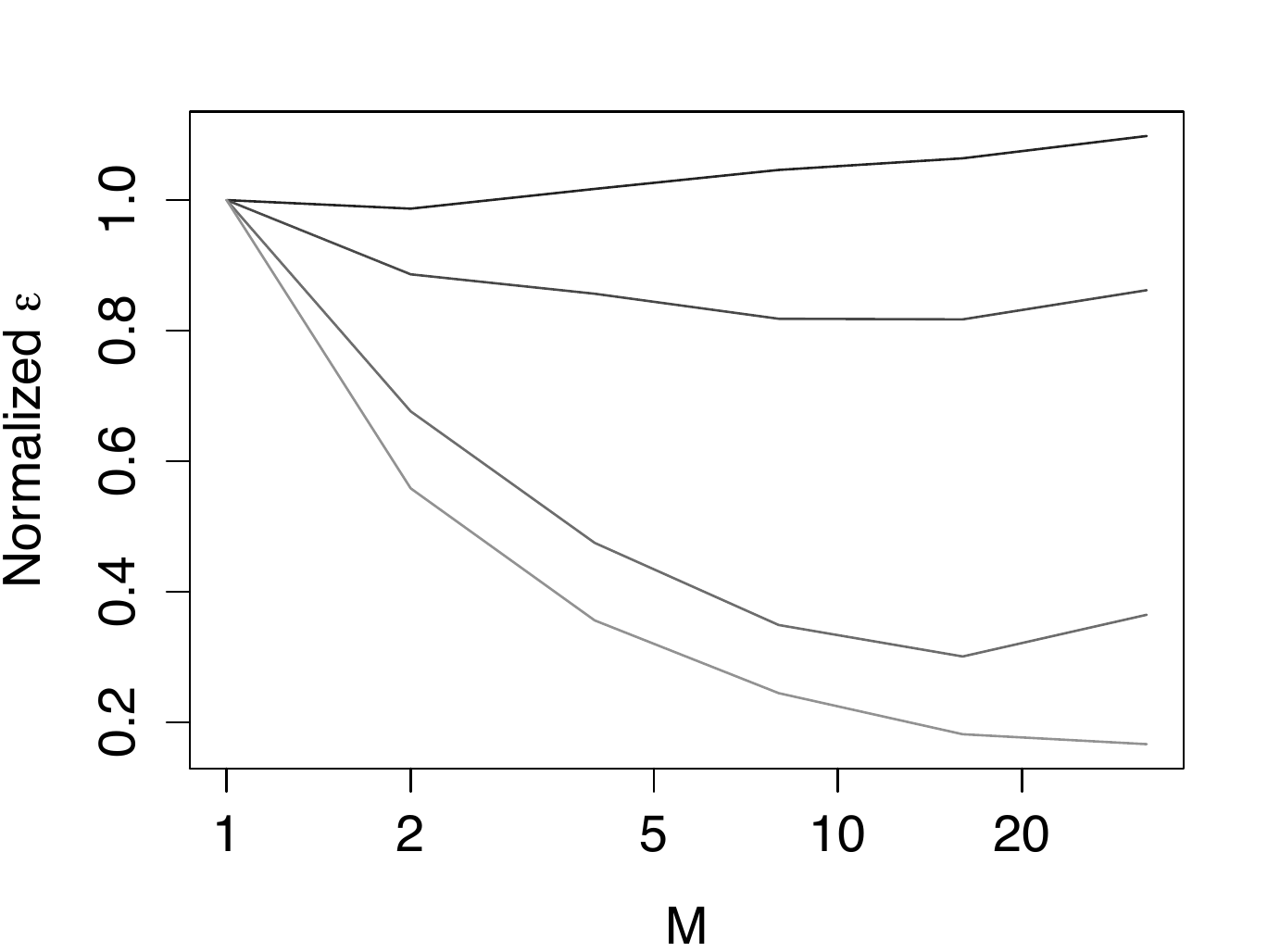}\qquad
\includegraphics[width=.45\textwidth]{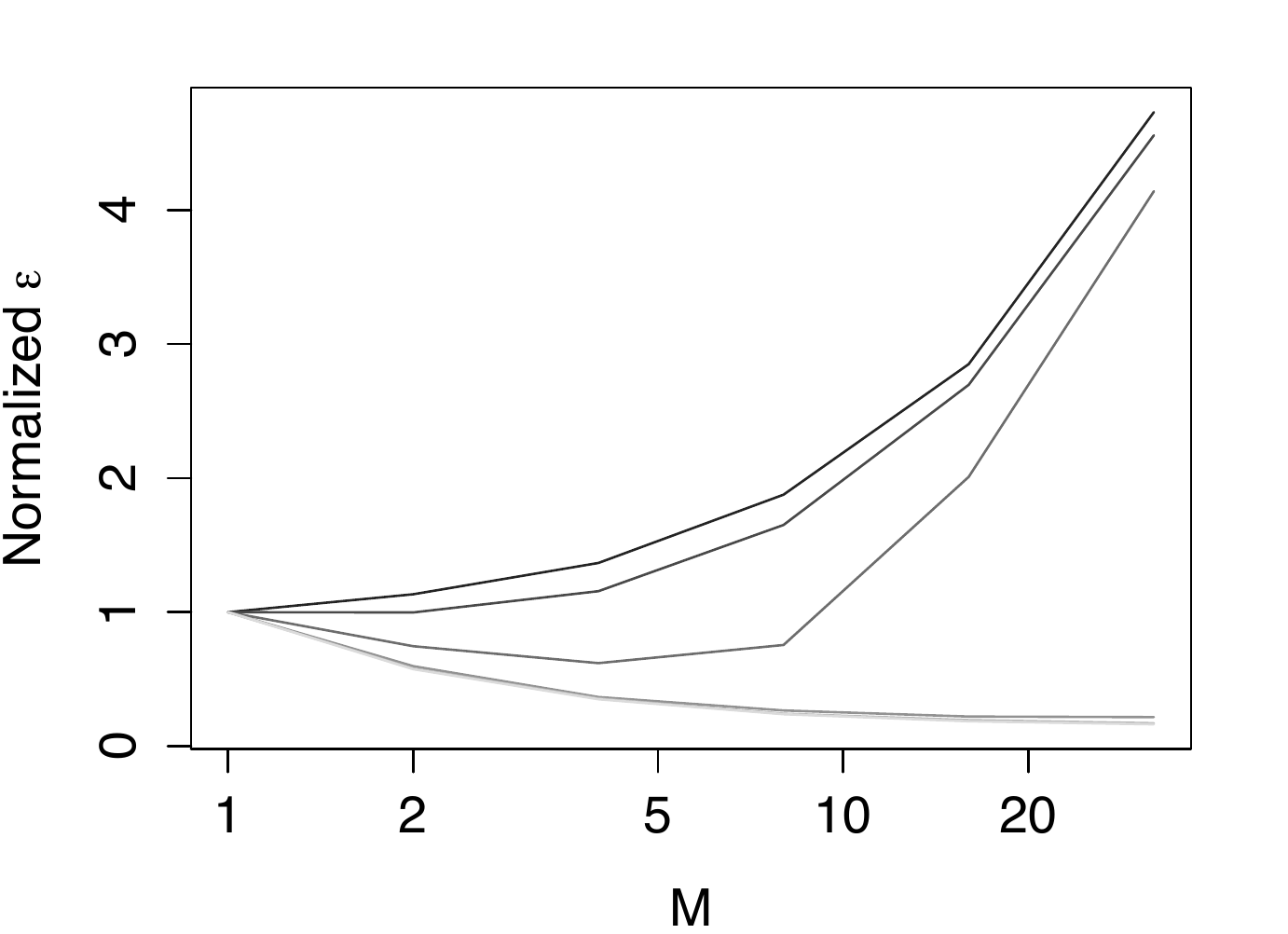}\\
\caption{\protect Sensitivity of the relationship between the bias and $M$ to
  changes in the likelihood. Each curve is normalized by dividing by its
  respective $\epsilon$ for $M=1$. Left: Varying $\dat = 2$ (light grey), $4,
  6, 8$ (black). Right: Varying $\sigma_y = 0.01$
  (black), $0.05, 0.1, 0.5, 1, 2$ (light grey).  \label{fig:3}}
\end{figure}
In summary, we only find a benefit of increasing the number of
pseudo-samples $M$ in cases where there is a discounted cost to obtain
those pseudo-samples, and even then the benefit can be decreased or
eliminated when $\dat$ is extreme under typical proposed values of
$\bt$, or when the variability of $\by$ under the model is low.

\section{Discussion}

In this paper, we have shown that despite the true likelihood leading to reduced asymptotic variance relative to the approximated likelihood constructed through ABC, in practice one should stick with simple, single pseudo-sample approximations rather than trying to accurately approximate the true likelihood with multiple pseudo-samples. Our results are obtained by bounding the asymptotic variance of Markov chain Monte Carlo estimates, which takes into account the autocorrelation of the Markov chain.  This is not to say that $M=1$ is optimal in all situations. In many cases, there is a large initial cost to the first pseudo-sample, with subsequent samples drawn at a much reduced computational cost. In this case $M>1$ can lead to improved performance.

We hope that this work not only provides practical guidance on the choice of the number of pseudo-samples when using ABC, but also that it might lead to future research in the analysis of these algorithms. As a specific example, it would be fruitful to pursue the results in this paper extended to non-uniform kernels such as the popular Gaussian kernel. The main difficulty in doing so is extending the explicit calculation \eqref{IneqMainCalcForProp} in the proof of Proposition \ref{thm:unifcor}. Although we have found it possible to obtain similar bounds for specific kernels, target distributions and proposal distributions via grid search on a computer, and simulations suggest that our conclusions hold in greater generality, we have not been able to extend this inequality to general kernels and general target distributions simultaneously.

\section*{Acknowledgement}
The authors thank Alex Thiery for his careful reading of an earlier draft, as well as Pierre Jacob, R\'{e}mi Bardenet, Christophe Andrieu, Matti Vihola, Christian Robert, and Arnaud Doucet for useful discussions. This research was supported in part by U.S.\ National Science Foundation grants 1461435, DMS-1209103, and DMS-1406599, by DARPA under Grant No. FA8750-14-2-0117, by ARO under Grant No. W911NF-15-1-0172, and by NSERC.

\appendix
\section{Proofs}

\begin{proof}[Proof of Theorem~\ref{thm:R1RM}]


Denote by $H_{2,\alpha}$ the transition kernel of the pseudo-marginal
algorithm with proposal kernel $q$, target marginal distribution
$\mu$, and estimator $T_{2,x,\alpha}$ of the unnormalized target. If we denote by $\{ (X_{t}^{(1)}, T_{t}^{(1)}) \}_{t \in \mathbb{N}}$ and $\{ (X_{t}^{(2)}, T_{t}^{(2)}) \}_{t \in \mathbb{N}}$ the Markov chains driven by the kernels $H_{2,\alpha}$ and $\alpha \mathrm{I} + (1-\alpha) H_2$ respectively, then $\{ X_{t}^{(1)} \}_{t \in \mathbb{N}}$ and $\{ X_{t}^{(2)} \}_{t \in \mathbb{N}}$ have the same distribution. 

If $T_{1,x} \leq_{cx}
 T_{2,x,\alpha}$ then by Theorem 3 of
 \cite{Andrieu2014establishing-some-order},
 \be \label{IneqViholaAppAsymVar}
v(f,H_{1}) &\leq v(f, H_{2,\alpha}) \ee
for any $f \in L^2(\mu)$.  We also have

\be \label{IneqLongCalc} v(f, H_{2,\alpha}) \leq \frac{1}{1-\alpha} v(f,H_{2}) + \frac{\alpha}{1 - \alpha} v(f) \leq \frac{1+
\alpha}{1 - \alpha} v(f, H_{2}),  \ee 
where the first inequality follows from Corollary 1 of \cite{latuszynski2013clts} and the second follows from the fact that $H_{2}$ has nonnegative spectrum. Combining this with
\eqref{IneqViholaAppAsymVar} yields the desired result. 

\end{proof}

\begin{proof}[Proof of Proposition~\ref{thm:unifcor}]
For any $M\geq 1$, let $T_{M,\bt}$ be the estimator
$\hat{\pi}_{K,M}(\bt|\dat)$ of the target $\pi_K$, so
that $T_{M,\bt,\alpha}$ is $T_{M,\bt}$ handicapped by $\alpha$ as
defined in \eqref{eqn:handicap} of the main document.  To obtain (\ref{eqn:firstBound}) of the main document, by
Theorem~\ref{thm:R1RM} it is sufficient to take $\alpha =
1-\frac{1}{M}$ and show that $T_{1,\bt} \leq_{cx} T_{M,\bt,\alpha}$.
By Proposition 2.2 of \cite{Lekela2014conditional-convex}, it is
furthermore sufficient to show that, for all $c \in \mathbb{R}$, \be
\label{IneqConvCond} \E[\vert T_{1,\bt} - c \vert] \leq \E[\vert T_{M,\bt,\alpha} - c
\vert].  \ee

Let $\mathrm{Bin}(n,\psi)$ denote the binomial distribution
with $n$ trials and success probability $\psi$. For a given point $\bt \in \Theta$, let $\tau = \tau(\bt) \equiv
\mathbb{P}[T_{1,\bt} \neq 0] = \int \textbf{1}_{\{\|\eta(\dat) -
  \eta(\by)\| < \epsilon\}} p(\by|\bt) d \by$. Noting that  $\frac{T_{1,\bt}}{\pi(\bt)} \in
\{0, 1 \}$, we may then write $T_{1,\bt}$
and $T_{M,\bt,\alpha}$ as the following mixtures
\be \frac{T_{1,\bt}}{\pi(\bt)} \stackrel{D}{=} \mathrm{Bin}(1,\tau), \qquad  \frac{T_{M,\bt,\alpha}}{ \pi(\bt)} \stackrel{D}{=} \frac{M-1}{M} \delta_{0} + \frac{1}{M}
\mathrm{Bin}(M,\tau),  \ee 
where $\delta_0$ is the unit point mass at zero.  Denote $T_{1,\bt}' =
\frac{T_{1,\bt}}{\pi(\bt)} $ and $T_{M,\bt,\alpha}' = \frac{T_{M,\bt,\alpha}}{\pi(\bt)}$. We will check condition \eqref{IneqConvCond} for
$T_{1,\bt}', T_{M,\bt,\alpha}'$ and $0 \leq c \leq 1$, then separately for $c < 0$
and $c > 1$. For $0 \leq c \leq 1$, we compute: 
\be \label{IneqMainCalcForProp} \E[\vert T_{M,\bt,\alpha}' -
c \vert] &= \left(1 - \frac{1}{M} \right) c + \frac{1}{M} (1 - \tau)^{M}
c + \frac{1}{M} \left( \sum_{j=1}^{M} \frac{M!}{j! (M-j)!} \tau^{j}
(1-\tau)^{M-j} \left( j - c \right) \right) \\ &= \tau + c \left( 1 -
\frac{2}{M}\left( 1 - (1-\tau)^{M} \right) \right) \geq \tau + c \left(1
- 2\tau \right)\qquad = \E[\vert T_{1,\bt}' - c \vert]. \ee For $c < 0$,
we have \be \E[\vert T_{M,\bt,\alpha}' - c \vert] = \E[T_{M,\bt,\alpha}'] - c = \E[T_{1,\bt}'] -
c = \E[\vert T_{1,\bt}' - c \vert],  \ee and the analogous calculation gives the same conclusion
for $c \geq M$. Finally, For $1< c < M$,  note
\be \label{IneqSimpleC} \E[\vert T_{1,\bt}' - 1 \vert] \leq \E[ \vert
T_{M,\bt,\alpha}' - 1 \vert], \qquad \E[\vert T_{1,\bt}' - M \vert] & = \E[\vert T_{M,\bt,\alpha}' -
M \vert].  \ee
Also, the functions $f_1(c) \equiv \E[\vert T_{1,\bt}' - c \vert]$ and
$f_2(c) \equiv \E[\vert T_{M,\bt,\alpha}' - c \vert]$ are continuous, convex and piecewise linear. For $c \geq 1$, they satisfy
\be \label{IneqDerC}
\frac{d}{dc} f_{1}(c) &= 1 \geq \frac{d}{dc} f_{2}(c)
\ee 
where the derivative of $f_{2}$ exists. Combining inequalities
\eqref{IneqSimpleC} and \eqref{IneqDerC}, we conclude that \be
\E[\vert T_{1,\bt}' - c \vert] \leq \E[\vert T_{M,\bt,\alpha}' - c
\vert] \ee for all $1 < c < M$. Thus we have verified
\eqref{IneqConvCond} and the proposition follows. 
\end{proof}

\section{Analysis of an Alternative ABC-MCMC Method}\label{sec:altmcmc}
We give a result analogous to Corollary~\ref{thm:likFree} for the
 version of ABC-MCMC proposed in \cite{wilk:08}, given in
Algorithm~\ref{likfree-mh} below.  The constant $c$ can be any value
satisfying $c \geq \sup_{\by} K(\eta(\dat) - \eta(\by))$.
\begin{algorithm}[h]
  Initialize $\bt^{(0)}$\;
  \For{$t=1$ to $n$}{
    Generate $\bm{\theta'} \sim q(\cdot \vert \bm{\theta}^{(t-1)})$, $\datt \sim p( \bm{y} | \bt')$, and $u \sim \mbox{Uniform}[0,1]$\;
    \eIf{$u \leq r(\bt^{(t-1)},\bt'|\datt) \equiv \frac{K(\eta(\dat) - \eta(\datt))}{c}
      \min\left\{1, \frac{\pi(\bt')q(\bt^{(t-1)} \vert \bt') }{\pi(\bt^{(t-1)})q(\bt' \vert \bt^{(t-1)}) } \right\}$,}{
        Set $\bm{\theta}^{(t)}$ = $\bm{\theta'}$
      }{
        Set  $\bm{\theta}^{(t)}$ = $\bm{\theta}^{(t-1)}$
      }
  }
\caption{Alternative ABC-MCMC Method\label{likfree-mh}
}
\end{algorithm}

Lemma~\ref{Thm:altmcmc} compares Algorithm~\ref{likfree-mh} (call its transition kernel $\tilde{Q}$) to $Q_\infty$.

\begin{lemma}\label{Thm:altmcmc}
For any $f\in L^2(\pi_K)$ we have $ v(f,
\tilde{Q}) \geq v( f, Q_{\infty})$.
\end{lemma}

\begin{proof}

Both $\tilde{Q}$ and $Q_\infty$ have stationary density $\pi_K$, so by Theorem 4 of \cite{tierney1998ordering}, 
it suffices to show that $Q_\infty(\bt,A\backslash \{\bt\}) \geq \tilde{Q}(\bt,A\backslash \{\bt\})$
for all $\bt \in \Theta$ and measurable $A \subset
\Theta$.  Since $\tilde{Q}$ and $Q_\infty$ use the same proposal density $q$, it is furthermore sufficient to show that for every
$\bt^{(t-1)}$ and $\bt'$, the acceptance probability of
$Q_\infty$ is at least as large as that of $\tilde{Q}$. Since $p(\cdot|\bt)$ is a probability density, 
\begin{align}\label{Eqn:supBound}
\int K(\eta(\dat)-\eta(\bm{y})) p(\bm{y}|\bt) d\bm{y} \leq \sup_{\bm{y}}
K(\eta(\dat)- \eta(\bm{y})) \qquad  \forall \bt \in \Theta.
\end{align}
So the acceptance probability
of $\tilde{Q}$, marginalizing over $\datt$, is
\be
a_{\mathrm{ABC}} &=  \int r(\bt^{(t-1)}, \bt'
| \by) p(\by | \bt') d \by\\
&\leq  \frac{\int K(\eta(\dat)-\eta(\by)) p(\by | \bt') d\by}{\sup_{\by} K(\eta(\dat)-\eta(\by))} 
\min\left\{1,\frac{\pi(\bt') q(\bt^{(t-1)} \vert \bt')}{\pi(\bt^{(t-1)}) q(\bt' \vert \bt^{(t-1)})}\right\}\\
&\leq   \min\left\{1,
\frac{\int K(\eta(\dat)-\eta(\by)) p(\by | \bt') d\by}{\sup_{\by}
  K(\eta(\dat)-\eta(\by))} \left(\frac{\pi(\bt') q(\bt^{(t-1)} \vert
  \bt')}{\pi(\bt^{(t-1)}) q(\bt' \vert
  \bt^{(t-1)})}\right)\right\}.  \label{eqn:accabcmc}
\ee
The acceptance probability of the Metropolis-Hastings algorithm is
\be
a_{\mathrm{MH}} = \min\left\{1,
\frac{\int K(\eta(\dat)-\eta(\by)) p(\by | \bt') d\by}{\int
  K(\eta(\dat)-\eta(\by)) p(\by | \bt^{(t-1)}) d\by}
\left(\frac{\pi(\bt') q(\bt^{(t-1)} \vert \bt')}{\pi(\bt^{(t-1)}) q(\bt' \vert
  \bt^{(t-1)})}\right)\right\}.  \label{eqn:accMH}
\ee
Using \eqref{Eqn:supBound}, \eqref{eqn:accabcmc}, and \eqref{eqn:accMH},
$a_{\mathrm{ABC}}/a_{\mathrm{MH}} \leq 1$. 
\end{proof}

\bibliographystyle{chicago}
\bibliography{abc}

\end{document}